\documentclass{article}

\usepackage{microtype}
\usepackage{graphicx}
\usepackage{subfigure}
\usepackage{booktabs} 

\usepackage{hyperref}
\usepackage{colortbl}
\usepackage{mathpazo}

\usepackage[accepted]{icml2024}
\usepackage{enumitem}
\usepackage{amsmath}
\usepackage{amssymb}
\usepackage{mathtools}
\usepackage{amsthm}

\usepackage[capitalize,noabbrev]{cleveref}

\theoremstyle{plain}
\newtheorem{theorem}{Theorem}[section]
\newtheorem{proposition}[theorem]{Proposition}
\newtheorem{lemma}[theorem]{Lemma}

\theoremstyle{definition}

\theoremstyle{remark}
\newtheorem{remark}[theorem]{Remark}

\usepackage[textsize=tiny]{todonotes}

\icmltitlerunning{DeepHoloBrain: From Geometric Deep Model to Computational Neuroscience}

\begin{document}

\twocolumn[
\icmltitle{Exploring the Enigma of Neural Dynamics Through A Scattering-Transform Mixer Landscape for Riemannian Manifold}



\icmlsetsymbol{equal}{*}

\begin{icmlauthorlist}
\icmlauthor{Tingting Dan}{yyy}
\icmlauthor{Ziquan Wei}{yyy,comp}
\icmlauthor{Won Hwa Kim}{sch}
\icmlauthor{Guorong Wu}{yyy,comp,kkk,zzz,zz}
\end{icmlauthorlist}

\icmlaffiliation{yyy}{Department of Psychiatry, University of North Carolina at Chapel Hill, Chapel Hill, NC, 27599, USA.}
\icmlaffiliation{comp}{Department of Computer Science, University of North Carolina at Chapel Hill, Chapel Hill, NC, 27599, USA.}
\icmlaffiliation{sch}{Computer Science and Engineering / Graduate School of AI, POSTECH, Pohang, 37673, South Korea.}
\icmlaffiliation{zzz}{UNC NeuroScience Center, University of North Carolina at Chapel Hill, Chapel Hill, NC, 27599, USA.}
\icmlaffiliation{kkk}{Department of Statistics and Operations Research (STOR), University of North Carolina at Chapel Hill, Chapel Hill, NC, 27599, USA.}
\icmlaffiliation{zz}{Carolina Institute for Developmental Disabilities, University of North Carolina at Chapel Hill, Chapel Hill, NC, 27599, USA}

\icmlcorrespondingauthor{Guorong Wu}{grwu@med.unc.edu}

\icmlkeywords{Machine Learning, ICML}

\vskip 0.3in
]

\printAffiliationsAndNotice{}  

\begin{abstract}
The human brain is a complex inter-wired system that emerges spontaneous functional fluctuations. In spite of tremendous success in the experimental neuroscience field, a system-level understanding of how brain anatomy supports various neural activities remains elusive. Capitalizing on the unprecedented amount of neuroimaging data, we present a physics-informed deep model to uncover the coupling mechanism between brain structure and function through the lens of data geometry that is rooted in the widespread wiring topology of connections between distant brain regions.
Since deciphering the puzzle of self-organized patterns in functional fluctuations is the gateway to understanding the emergence of cognition and behavior, we devise a geometric deep model to uncover manifold mapping functions that characterize the intrinsic feature representations of evolving functional fluctuations on the Riemannian manifold. In lieu of learning unconstrained mapping functions, we introduce a set of graph-harmonic scattering transforms to impose the brain-wide geometry on top of manifold mapping functions, which allows us to cast the manifold-based deep learning into a reminiscent of \textit{MLP-Mixer} architecture (in computer vision) for Riemannian manifold. As a proof-of-concept approach, we explore a neural-manifold perspective to understand the relationship between (static) brain structure and (dynamic) function, challenging the prevailing notion in cognitive neuroscience by proposing that neural activities are essentially excited by brain-wide oscillation waves living on the geometry of human connectomes, instead of being confined to focal areas.
\end{abstract}

\section{Introduction}
\label{intro}
Human brain is a complex system physically wired by massive bundles of nerve fibers \cite{bassett2017network}. On top of intertwined structural connectomes, ubiquitous neural oscillations emerge remarkable functional fluctuations, synchronizing across large-scale neural circuits, that support myriad high-level cognitive functions necessary for everyday living \cite{bressler2010large}. With the prevalence of structural and functional neuroimaging technology (aka. Magnetic Resonance Imaging, MRI) in many neuroscience studies, the idea of understanding the mind forms an important concept that the dynamic nature of human brain cannot be understood by thinking of the system as comprised of independent components \cite{terras2012harmonic}. In this regard, there is a critical need to establish a system-level understanding of how the brain function emerges from anatomical structures and how the self-organized system behavior of functional fluctuations supports the cognitive states.

Like many dynamic systems in the universe, the evolving functional fluctuations manifest remarkable geometric patterns to the extent of self-organized spontaneous co-activation of neural activities. As shown in Fig. \ref{fig:overview} (left), functional connectivity (FC), formed by the pairwise correlation between two time courses of BOLD (blood-oxygen-level-dependent) signals \cite{bullmore2009complex}, exhibits \textit{small world} properties \cite{watts1998collective}. From a data science perspective, the FC matrix lives on the high-dimensional Riemannian manifold of SPD (symmetric and positive-definite) matrices. Due to many well-studied mathematical properties of SPD matrix, tremendous efforts have been made to statistical modeling \cite{you2021re}, clustering \cite{qiu2015manifold}, and characterization of temporal dynamics \cite{dan2022learning} from the manifold instances of SPD matrices.

Although manifold-based deep learning often presents greater challenges compared to the counterpart machine learning backbones operating in Euclidean space, prioritizing the preservation of data geometry substantially improves both the reliability of model explanations and overall accuracy \cite{tiwari2022effects}. In this regard, various manifold-based deep models have been proposed for learning feature representation of SPD matrices using Riemannian manifold algebras \cite{huang2017riemannian, dan2022uncovering}. As shown in Fig. \ref{fig:overview}, the driving factor of most deep models for SPD matrices is to find a cascade of mapping functions that progressively project the input SPD matrices into a latent subspace of Riemannian manifold which is in line with the downstream learning tasks. Since there is no constraint on the mapping functions, every element in the mapping function is \textit{independently} updated through the gradients in back-propagation. Thus, it is evident that (1) the learned mapping functions lack explainability, (2) such deep model may demand a relatively large volume of training data, and (3) adapting a pre-trained model to a new dataset may pose challenges.
\begin{figure}[h]
    \centering
    \vspace{-0.5em}
    \includegraphics[width=0.48\textwidth]{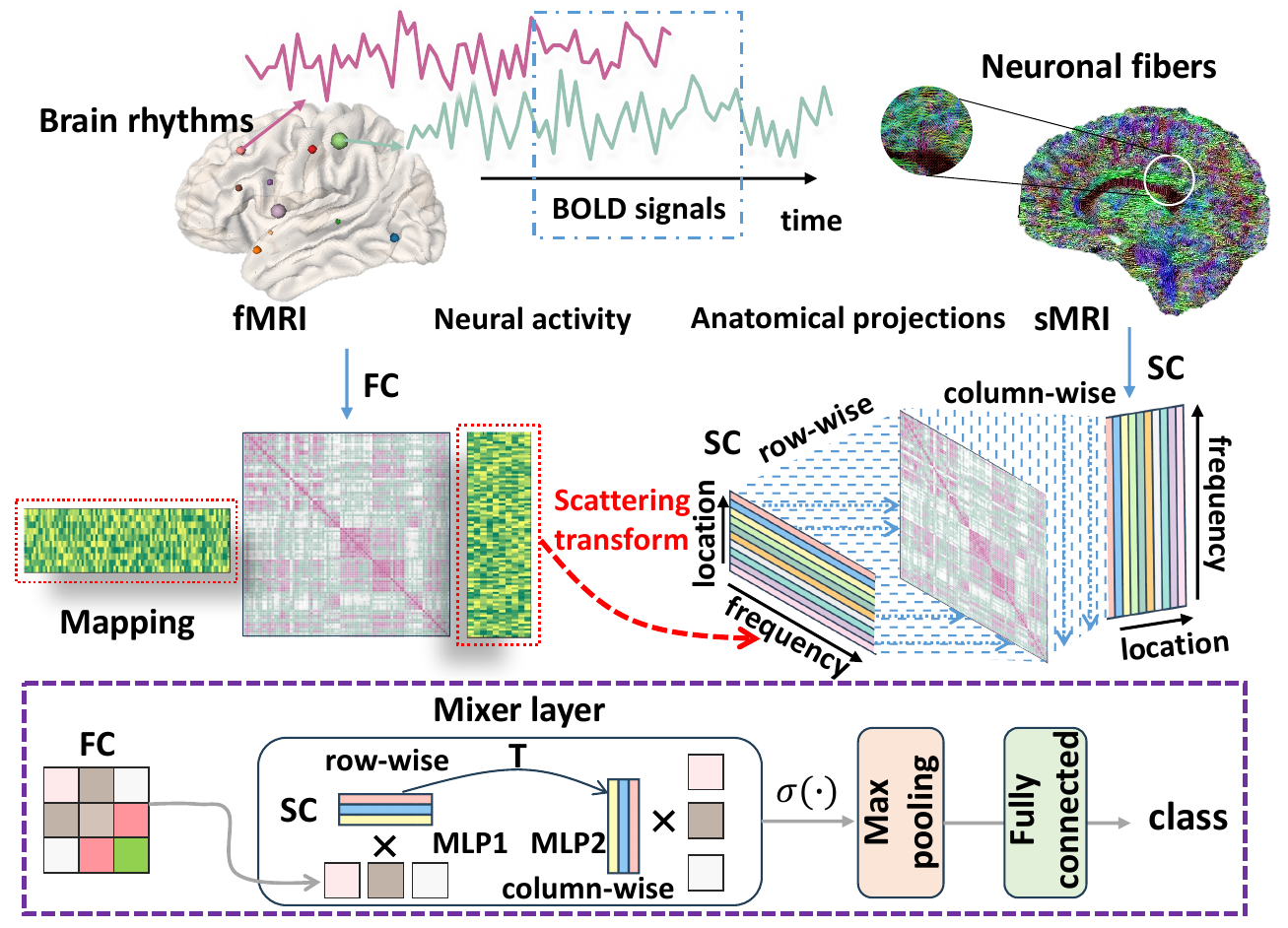}
    \vspace{-1.7em}
    \caption{Overview of answering open neuroscience questions using machine learning techniques. \textit{Top}: The motivation of our work is to understand how brain structure supports ubiquitous functional fluctuations. \textit{Middle}: Compared to conventional mapping functions (left) in manifold-based deep learning, we introduce scattering transform to form the basis of mapping functions on the Riemannian manifold (right), which allows us to constrain the mapping of functional connectivities underlying the geometric of structural region-to-region connectome. \textit{Bottom}: We devise a Mixer architecture on top of the row-wise and column-wise scattering transform which not only yields a novel geometric deep model for Riemannian manifold but also provides a new neuroscience insight to understand the coupling mechanism between brain structure and function.}
    \vspace{-1.5em}
    \label{fig:overview}
\end{figure}

To address these limitations, we introduce the notion of scattering transform to impose a geometric constraint on top of the mapping functions in the Riemannian manifold. Following the notion of recent pioneering work on brain structural and functional coupling \cite{pang2023geometric}, each scattering transform operator is derived from the graph spectrum of structural connections, essentially operating as harmonic wavelets at various oscillation frequencies \cite{hammond2018spectral}. To link manifold-based deep learning with scattering transforms, we present a novel mapping mechanism on the Riemannian manifold that allows us to preserve the data geometry in machine learning. Furthermore, we dissect the mapping process on Riemannian manifold into a row-wise and another column-wise scattering transforms (Fig. \ref{fig:overview}, middle-right) on each SPD matrix, which sets the stage for the learning scenario of Mixer architecture, a reminiscent of \textit{MLP-Mixer} for computer vision \cite{tolstikhin2021mlp}. Due to the simple network architecture, our scattering-transform Mixer on Riemennian manifold of SPD matrices can be trained efficiently compared to existing manifold-based deep models with a significantly reduced number of parameters. Meanwhile, graph-harmonic scattering transforms offer a new window to interpret the brain-wide contribution of mapping functions on the coupling mechanism between brain structure and function.

Our work makes a significant contribution to the forefront of computational neuroscience, enriching traditional experimental neuroscience through the novel insight into brain structure-function coupling mechanism from the deep model. Specifically, we put the spotlight on the physics principle of scattering-transform in the Mixer, which acts as a stepping stone to understanding the mechanistic role of brain structure on functional fluctuations. We use the geometric pattern of graph harmonics to show that self-organized neural dynamics can be parsimoniously understood as resulting from the wave-to-wave interference formed by superimposing the fundamental resonance modes of the brain's geometry (i.e., topology of region-to-region connectivities) on the subject-specific neural activities.

Furthermore, we connect the interference phenomenon discovered in our deep model to the holography technique \cite{holograph_gabor} in computer vision, a stereo-imaging technique that generates a hologram by superimposing a reference beam on the wavefront of interest. Following this notion, we present a proof-of-concept theory, called \textit{DeepHoloBrain}, which computationally “records” the cross-frequency couplings (CFC) of time-evolving interference patterns where spontaneous functional fluctuations are constrained by the graph harmonic waves (with predefined oscillation frequencies).

In practice, we evaluate the clinical value of our manifold mixer model in (1) disease early diagnosis and (2) generality of adapting the pre-trained model to a new dataset. Compared with current state-of-the-art deep models, our method not only achieves the best performance in terms of accuracy and consistency but also shows potential in addressing real-world challenges in clinical practice, particularly in situations where disease data cohorts often have very limited sample sizes.

Together, the contribution of our work has four folds: (1) a new mapping mechanism on the Riemannian manifold constrained by scattering transforms, (2) a novel geometric deep model of scattering-transform Mixer for Riemannian manifold of SPD matrices, (3) a new gateway to understanding the human brain using deep learning technique, and (4) a comprehensive evaluation of the clinical impact on large-scale neuroimages.

\section{Background and Related Works}
\subsection{Open Questions in Network Neuroscience}
In the realm of neuroscience, a plethora of studies investigate connections between $N$ predefined regions in the brain \cite{bassett2017network}. Through diffusion-weighted imaging (DWI), we are able to measure the structural connections (SC) between two brain regions $\Omega_i$ and $\Omega_j$ $(i,j=1,...,N)$, \textit{in vivo}, by tracking bundles of nerve fibers. On the other hand, functional magnetic resonance imaging (fMRI) technique enables to capture spontaneous functional connections (FC) which indicate the synchronization between two time courses of neural activities at different brain regions $\Omega_i$ and $\Omega_j$. In general, we can form a network of SC or FC, where each element denotes the region-to-region SC/FC connection. A brief summary of an analytic pipeline for constructing SC and FC networks is shown in Appendix \ref{process}.
In contrast to the diverse FC patterns associated with cognition and behavior, SC remains relatively static. In this context, several interesting scientific questions arise: \textit{How does the structural foundation of the brain shape its dynamic functional activities? And, how does the coupling between SC and FC contribute to the emergence of cognition and behavior?}

Since each FC matrix is symmetry and positive definite (SPD), striking efforts have been made to understand the self-organized pattern of FC through the lens of Riemannian manifold of SPD matrices \cite{you2021re,dan2022uncovering}. Following this spirit, we seek to explore the SC-FC coupling mechanism using geometric deep models on the Riemannian manifold, where the mathematical insight might open a new window to answer neuroscience questions through machine learning.

\subsection{Canonical Deep Model on SPD Matrices} 
\label{spdnet}
Suppose we have a $N\times N$ SPD matrix $\mathbf{X}$ residing on Riemannian manifold, i.e., $\mathbf{X} \in Sym^+_d (d=N)$, current deep models such as SPDNet \cite{huang2017riemannian} are trained to find a set of non-linear mapping functions to obtain a new feature representation $\mathbf{X}_{l}$ at $l^{th}$ layer by:
\begin{equation}
\vspace{-0.5em}
    \mathbf{X}_l=\mathbf{\Psi}_l\mathbf{X}_{l-1}\mathbf{\Psi}_l^\intercal ,
    \label{eq:SPDNet}
\end{equation}
where the learnable mapping function $\mathbf{\Psi}_l \in \mathbb{R}^{d_l \times d_{l-1}}$ is a row full-rank matrix (usually $d_l \ll d_{l-1}$). $\mathbf{\Psi}_l$ is called positive mapping in Riemannian manifold of SPD matrices \cite{bhatia2009positive}, which ensure the output $\mathbf{X}_l$ maintains the geometric property of SPD matrix, i.e, $\mathbf{X}_l \in Sym^+_{d_{l}}$. Although it is a common practice to impose column-wise orthogonality on $\mathbf{\Psi}_l$, the optimization of each element in $\mathbf{\Psi}_l$ is completely independent in back-propagation. As shown in the middle-left of Fig. \ref{fig:overview}, the learned mapping function $\mathbf{\Psi}_l$ lacks interpretability in the context of data geometry. To address this limitation, we introduce the scattering transform technique to decompose each $\mathbf{\Psi}_l$ into two dimensions: frequency and location (the middle-right of Fig. \ref{fig:overview}).

\subsection{Scattering Transforms on Graph} The scattering transform (such as wavelets) \cite{mollai2010recursive, mallat2012group,gama2019stability} offers a stable and multi-scale representation for signals in the Euclidean space $\mathbb{R}^d$. A collection of pre-defined scattering transforms can decompose any signals into scale and orientation components separately, which allows us to capitalize on geometric intuition while maintaining useful properties for machine learning such as translational invariance. Following this spirit, the scattering transform technique has been extended to geometric data living on the irregular domains (such as graph and manifold) \cite{perlmutter2018geometric,perlmutter2020geometric,gao2019geometric}.

Consider a smooth, compact, and connected $N$-dimensional Riemannian manifold $\mathcal{M}$ and a square-integrable function $f\in \mathbf{L}^2(\mathcal{M})$. The negative Laplace-Beltrami operator $-\Delta$ on $\mathcal{M}$ possesses a countable set of eigenvalues $0 = \lambda_0 < \lambda_1 \leq \lambda_2 \leq \ldots \leq \lambda_{K-1}$. There exists a corresponding sequence of eigenfunctions ${u}_k (k=0,...,K-1)$ satisfying $-\Delta {u}_k = \lambda_k {u}_k$, where $\{{u}_k\}_{k \geq 0}$ forms an orthonormal basis for the manifold $\mathcal{M}$. Following the notion in \cite{perlmutter2020geometric,perlmutter2018geometric}, the eigenfunctions $\{{u}_k\}_{k \geq 0}$ serve as generalized Fourier modes on the manifold $\mathcal{M}$ that is spanned by the inner product with these eigenfunctions. In this context, any function $f$ can be represented as a sum over these eigenfunctions as $f(x) = \sum_{k \geq 0} \widehat{f}(k)u_k(x)$,
where $\widehat{f}(k)=\left\langle f, {u}_k\right\rangle_{\mathbf{L}^2(\mathcal{M})}=\int_{\mathcal{M}} f(y)\overline{{u}_k(y)} dy$ is known as Fourier coefficient corresponding to the eigenfunction $u_k$.

\begin{figure*}[h]
    \centering
    \vspace{-0.5em}
    \begin{minipage}{0.7\textwidth}
        \centering
        \includegraphics[width=0.88\linewidth]{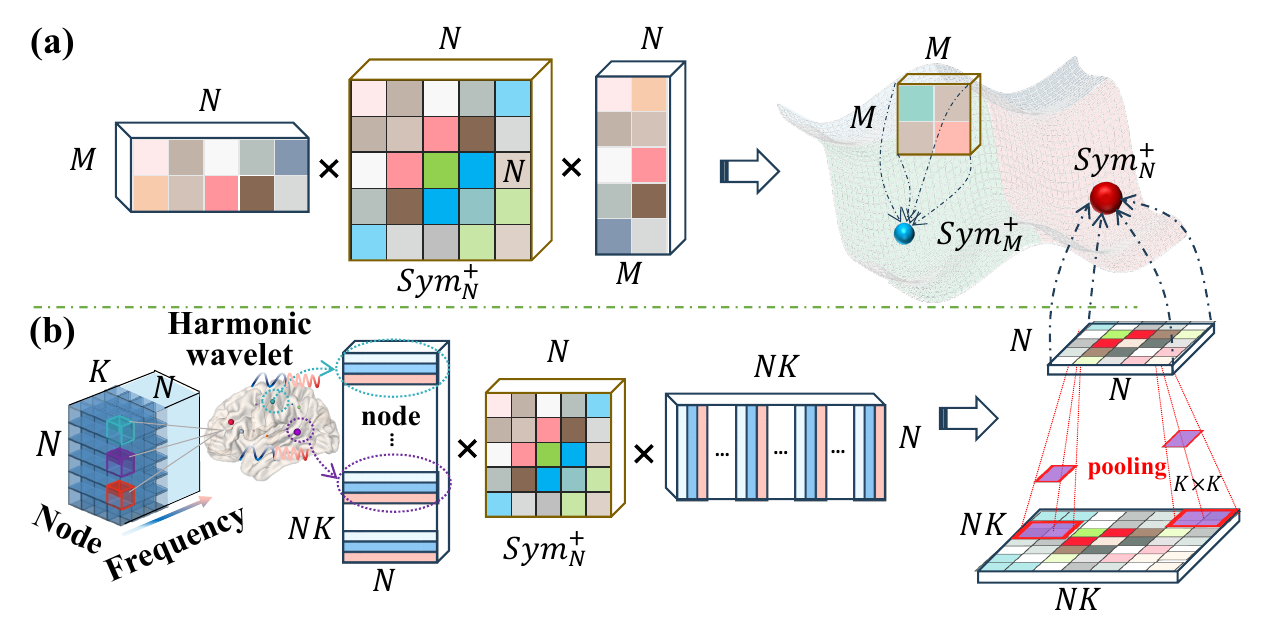}
    \end{minipage}%
    \begin{minipage}{0.3\textwidth}
        \centering
        \vspace{-1.7em}
        \caption{(a) Conventional mapping function (Eq. \ref{eq:SPDNet}) on Riemannian manifold of SPD matrices, where each element in the mapping matrix is estimated independently. (b) Our new mapping function which consists of (1) constructing the block matrix $\mathbb{P}$ by stacking harmonic wavelets over frequency and location, (2) yielding a \textit{Supra-FC} matrix $\mathbb{X}=\mathbb{P}^\intercal \boldsymbol{X}\mathbb{P}$ by Eq. \ref{eq:SPDNet}, and (3) applying pooling operation on $\mathbb{X}$ across frequencies at each location. By doing so, the output is a $N\times N$ SPD matrix.}
        \label{fig2}
    \end{minipage}
     \vspace{-1.8em}
\end{figure*}

Furthermore, we can expand Fourier's bases to the wavelet scattering transform on manifold $\mathcal{M}$ by localizing each eigenfunction $u_k$ to particular region $\Omega_i (i=1,...,N)$, yielding a frequency-specific and location-specific harmonic wavelet function $\psi_i^{k}$. In this work, we construct harmonic wavelets on top of the graph topology \cite{hammond2011wavelets,shuman2013emerging}. Suppose the underlying graph has $N$ nodes and the node-to-node connection is encoded in a $N\times N$ adjacency matrix $\mathbf{A}$. Thus, the graph Laplacian matrix $\mathcal{L}$ can be constructed by $\mathbf{\mathcal{L}}=\mathbf{D}-\mathbf{A}$, where $\mathbf{D}\in \mathbb{R}^{N\times N}$ is a diagonal matrix of node-wise total connectivity degree. Next, we apply eigen-decomposition on $\mathbf{\mathcal{L}}$, yielding $\mathbf{\mathcal{L}}=\boldsymbol{U\Lambda U}^\intercal$, where $\boldsymbol{\Lambda}=diag[\lambda_k]_{k=1}^N$ is a diagonal matrix of eigenvalues and $\boldsymbol{U}=[u_k]_{k=1}^N$ are eigenvectors, each corresponding to the eigenvalue $\lambda_k$. For each eigenvector $u_k$, we can derive in total $N$ harmonic wavelets by localizing $u_k$ to the underlying region $\Omega_i$ by \cite{shuman2020localized}:
\begin{equation}
    \psi_i^k=\boldsymbol{U}g_k(\gamma\boldsymbol{\Lambda})\boldsymbol{U}^\intercal\delta_i
    \label{eq.wavelets}
\end{equation}
where $\delta_i$ is a graph signal with a value of one at region $\Omega_i$ and zero elsewhere. $g_k(\cdot)=e^{-\gamma \lambda_k}$ is a spectral filter that defines the spectral pattern that is localized to each region $\Omega_i$, where $\gamma$ is a learnable hyper-parameter that controls the scale of spectral pattern.

\section{Methods}
\subsection{New Mapping Function on Riemannian Manifold Constrained by Graph Scattering Transforms}
Suppose both SC and FC networks consist of $N$ nodes. We construct $K (K\leq N)$ harmonic wavelets at each node $\Omega_i$ of SC network, denoted by $\mathbf{\Psi}_i = [\psi_i^k]^K_{k=1}$, where $k$ indicates harmonic frequency. For convenience, we use $\mathbb{P}=[[\boldsymbol{\Psi}_1], [\boldsymbol{\Psi}_2],...,[\boldsymbol{\Psi}_N]]\in \mathbb{R}^{N\times NK}$ to denote the block matrix by stacking $N$ region-specific matrices of harmonic wavelets $\boldsymbol{\Psi}_i$ one after another, as shown in Fig. \ref{fig2} (b). Next, we present a new mapping function, i.e., applying scattering transforms $\mathbb{P}$ on FC matrix $\boldsymbol{X}$, yielding a new SPD matrix on the Riemannian manifold $\mathcal{M}$.

\begin{remark}
\label{SPD}
A real and symmetric matrix $\boldsymbol{X} \in Sym^+_N $ is said to be SPD matrix if $v\boldsymbol{X}v^\intercal = \sum_{i,j=1}^N v_ix_{ij}v_j > 0$ for any non-zero vector in $v \in \mathbb{R}^N$.
\end{remark}

\begin{lemma}
    Given a full-rank matrix $ \boldsymbol{W} \in \mathbb{R}^{N \times M} (M<N)$ and a vector $ v \in \mathbb{R}^M $, since the columns of $ \boldsymbol{W} $ are linearly independent and span $ \mathbb{R}^M $, there exists, by the fundamental theorem of linear algebra, a vector $ u \in \mathbb{R}^N $ such that $ \boldsymbol{W} u = v $ (Proof shown in Proposition \ref{theorem:lowrank2}).
    \label{theorem:lowrank}
\end{lemma}

\begin{remark}
\label{def:inj}
Given an SPD matrix $\boldsymbol{X} \in Sym^+_N$ and a full-rank matrix $\boldsymbol{W} \in \mathbb{R}^{N \times M} (M<N)$, the function $f : Sym^+_N \times \mathbb{R}^{N \times M} \rightarrow Sym^+_M$ defined as $f(\boldsymbol{X}, \boldsymbol{W}) = \boldsymbol{W}^\intercal \boldsymbol{X} \boldsymbol{W}$ \cite{huang2017riemannian}.
\end{remark}

\begin{proposition}
    The output of $f(\boldsymbol{X}, \boldsymbol{W})$ is a $N\times N$ SPD matrix.
    \label{prop:SPD}
\end{proposition}
\vspace{-1.5em}
\begin{proof}
According to Lemma \ref{theorem:lowrank}, there exists a vector $u \in \mathbb{R}^N$ such that $\boldsymbol{W} u = v$. Note that $u$ is not the zero vector because $v$ is not zero and $\boldsymbol{W}$ has full rank.
Substitute $v = \boldsymbol{W} u$ into the quadratic form: $v^\intercal (\boldsymbol{W}^\intercal \boldsymbol{X} \boldsymbol{W}) v = (\boldsymbol{W} u)^\intercal (\boldsymbol{W}^\intercal \boldsymbol{X} \boldsymbol{W}) (\boldsymbol{W} u) = u^\intercal \boldsymbol{W} \boldsymbol{W}^\intercal \boldsymbol{X} \boldsymbol{W} \boldsymbol{W}^\intercal u$.
Since $\boldsymbol{W}$ has full rank and $M < N$, $\boldsymbol{W} \boldsymbol{W}^\intercal$ is positive definite. As the product of positive definite matrices $\boldsymbol{W}\boldsymbol{W}^\intercal$, $\boldsymbol{X}$, and $\boldsymbol{W}\boldsymbol{W}^\intercal$ yields another positive definite matrix, we have $u^\intercal (\boldsymbol{W} \boldsymbol{W}^\intercal \boldsymbol{X} \boldsymbol{W} \boldsymbol{W}^\intercal) u > 0$. Because $u$ is not the zero vector, the quadratic form is strictly positive. This proves that the output of $f(\boldsymbol{X}, \boldsymbol{W}) = \boldsymbol{W}^\intercal \boldsymbol{X} \boldsymbol{W}$ is indeed an SPD matrix (as shown in Fig. \ref{fig2} (a)).
\end{proof}

\begin{proposition}
    If the number of rows in $\boldsymbol{W}$ is less than the number of columns, i.e, $N<M$, the mapping $f(\boldsymbol{X}, \boldsymbol{W}) = \boldsymbol{W}^\intercal \boldsymbol{X} \boldsymbol{W}$ only yields a semi-positive definiteness matrix.
    \label{prop:semi-SPD}
\end{proposition}
\vspace{-1.5em}
\begin{proof}
    Similar to the proof for Proposition \ref{prop:SPD}, we examine the quadratic form given by $ v^\intercal (\boldsymbol{W}^\intercal X \boldsymbol{W}) v = (\boldsymbol{W}v)^\intercal X (\boldsymbol{W}v)$ for any vector $v$. Since $N<M$, there exist a vector $v$ such that $\boldsymbol{W}v=0$, resulting in $f=(\boldsymbol{X}, \boldsymbol{W})=0$. Since $\boldsymbol{X}$ is SPD, $ (\boldsymbol{W}v)^\intercal \boldsymbol{X} (\boldsymbol{W}v) > 0 $ as long as $\boldsymbol{W}v$ is non-zero vector. Together, we have proven Proposition \ref{prop:semi-SPD}.
\end{proof}

Recall that $\mathbb{P}\in \mathbb{R}^{N\times NK}$ is a block matrix (primarily indexed by location $\Omega_i$), where each block element is a $N\times K$ matrix. Since $N<NK$ (row number less than column number), $\Tilde{\mathbb{P}}=f(\boldsymbol{X},\mathbb{P})$ does not guarantee that the output matrix resides on the Riemannian manifold of SPD matrices. However, $\mathbb{X}=\mathbb{P}^\intercal \boldsymbol{X}\mathbb{P}\in \mathbb{R}^{NK\times NK}$ essentially expand each (scalar) element in $\boldsymbol{X}$ to a $K\times K$ matrix. In this context, we coin $\mathbb{X}$ as \textit{Supra-FC} matrix, where each element $\mathbb{X}_{ij}$ is a $K\times K$ block matrix, i.e., $\mathbb{X}_{ij}=[\psi_i^s\boldsymbol{X}\psi_j^t]_{s,t=1}^K$. In light of this, we seek to use the following pooling operation to regain the SPD property while reducing the dimensionality.

\begin{proposition}
Given a \textit{Supra-FC} matrix $\mathbb{X}$, we apply max-pooling operation to each block matrix $\mathbb{X}_{ij}$ (over frequencies $s$ and $t$). Thus, the resulting $N\times N$ matrix

$ \Tilde{\boldsymbol{X}} = \begin{bmatrix} \text{max}(\mathbb{X}_{11}) & \text{max}(\mathbb{X}_{12}) & \cdots & \text{max} (\mathbb{X}_{1N}) \\ \text{max}(\mathbb{X}_{21}) & \text{max}(\mathbb{X}_{22}) & \cdots & \text{max}(\mathbb{X}_{2N}) \\ \vdots & \vdots & \ddots & \vdots \\ \text{max}(\mathbb{X}_{N1}) & \text{max}(\mathbb{X}_{N2}) & \cdots & \text{max}(\mathbb{X}_{NN}) \end{bmatrix} $

is positive-definite.

\label{SPD2}
\end{proposition}
\vspace{-1.5em}
\begin{proof}
\label{spd_proof}
We sketch the proof as follows. By substituting the formulation of wavelets into each $\psi_i^s\boldsymbol{X}\psi_j^t$, we have
$\mathbb{X}_{ij}[s,t]=g(\lambda_s)g(\lambda_t)u_s(i)u_s(j)\boldsymbol{X}u_t(i) u_t(j)$
(for convenience, we skip the scaling effect $\gamma$ and index value $\delta$ in Eq. \ref{eq.wavelets}). Next, we prove that the lower bound of
$max(u_s(i)u_s(j)\boldsymbol{X}u_t(i) u_t(j))$ is $\frac{1}{{N}}$
in Proposition \ref{laplacian} of Appendix. Recall that $g(\cdot)$ is an exponential function, we show that $\Tilde{X}$ is positive-definite in Proposition \ref{spd_prove1} of Appendix since $\Tilde{X}=e^B$ where $B$ is a $N\times N$ matrix.
\vspace{-1.0em}
\end{proof}

\textbf{Remarks.} In Fig. \ref{fig2}(b), we demonstrate the workflow of new mapping function on the Riemannian manifold which includes (1) constructing \textit{Supra-FC} matrix $\mathbb{X}=\mathbb{P}^\intercal\boldsymbol{X}\mathbb{P}$, (2) applying block-wise max-pooling operation $\boldsymbol{X}=\text{max-pooling}(\mathbb{X})$ via Proposition \ref{SPD2}, and (3) make $\Tilde{X}$ symmetric by $\Tilde{X}\leftarrow \frac{1}{2}(\Tilde{X}+\Tilde{X}^\intercal)$ (refer to Proposition \ref{spd_prove2} in Appendix). Recall that there is a learnable parameter $\gamma$ in Eq. \ref{eq.wavelets}. In this regard, we present the following deep model to optimize scattering transforms in $\mathbb{P}$ through the lens of scaling effect on harmonic wavelets.

\subsection{\textit{MLP-Mixer} for Riemannian Manifold}
Since each row in $\mathbb{P}$ is associated with a harmonic wavelet, the operation of $\mathbb{P}^\intercal\boldsymbol{X}\mathbb{P}$ can be boiled down to first apply scattering transform to each column of $\boldsymbol{X}$ followed by same operation on its rows. Inspired by the efficient \textit{MLP-Mixer} architecture in computer vision \cite{tolstikhin2021mlp}, we propose to deploy a set of MLPs to learn the optimal scaling effect $\gamma$ in the column-wise scattering transforms $\mathbb{P}X$ and then another set of MLPs for the optimal scaling effect $\gamma$ in the counterpart row-wise scattering transforms. Suppose $\sigma_{row}$ and $\sigma_{column}$ denote the MLPs (followed by the ReLU layer on manifold \cite{huang2017riemannian}) for row-wise and column-wise scattering transforms, respectively. The \textit{MLP-Mixer} architecture of positive mapping $\mathbb{P}^\intercal X\mathbb{P}$ on Riemannian manifold can be formulated as:
\begin{equation}
\label{forward}
\left\{
\begin{aligned}
    &\mathbb{X}_{l+1} = \sigma_{row}\{\text{LN}[\sigma_{column}(\text{LN}(\mathbb{P}^\intercal\Tilde{X}_l))\mathbb{P}]\}\\
    &\Tilde{X}_{l+1} = \text{max-pooling}(\mathbb{X}_{l+1})\\
    &\Tilde{X}_{l+1} = \frac{1}{2}(\Tilde{X}_{l+1}+(\Tilde{X}_{l+1})^\intercal)
\end{aligned}
\right.
\end{equation}
where LN($\cdot$) denotes LayerNorm.

\textbf{Remarks.} It is clear that the input to each \textit{MLP-Mixer} is a SPD matrix $\Tilde{X}_l$. The backbone of positive mapping on the Riemannian manifold is nothing but a set of MLPs and block-wise max-polling operations, yielding a new SPD matrix $\Tilde{X}_{l+1}$ in the same dimensions. In our implementation, we employ the classic positive mapping function in Eq. \ref{eq:SPDNet} to learn the low-dimensional feature representations. After that, we project the learned SPD matrix to the tangent space \cite{huang2015log} and then connect to the fully-connected layer for various machine learning applications such as classification.

\subsection{\textit{DeepHoloBrain}: A Proof-of-Concept Approach to Explore the Enigma of Neural Dynamics Through the Insight of Deep Model}
\label{DeepHoloBrain}
\textbf{Scattering Transforms: Stepping Stone between Deep Learning and SC-FC Couplings.} In functional MRI studies, the FC matrix $\boldsymbol{X}$ encodes pairwise correlations between two time course of neural activities $h_i=[h_i(1),h_i(2),...,h_i(T)]$ at region $\Omega_i$ and $h_j(t)$ at region $\Omega_j$, where $T$ is the number of time points. For convenience, we stack each $h_i$ row-by-row and form a whole-brain time course matrix $\boldsymbol{H}\in \mathbb{R}^{N\times T}$ (purple lines in Fig. \ref{snapshot}). In this context, the FC matrix can be computed by $\boldsymbol{X}=\boldsymbol{H}\boldsymbol{H}^\intercal$ (brown box), where each element $x_{ij}=h_ih_j^\intercal$.

\begin{figure}[h]
    \centering
    \vspace{-1.2em}
    \includegraphics[width=0.4\textwidth]{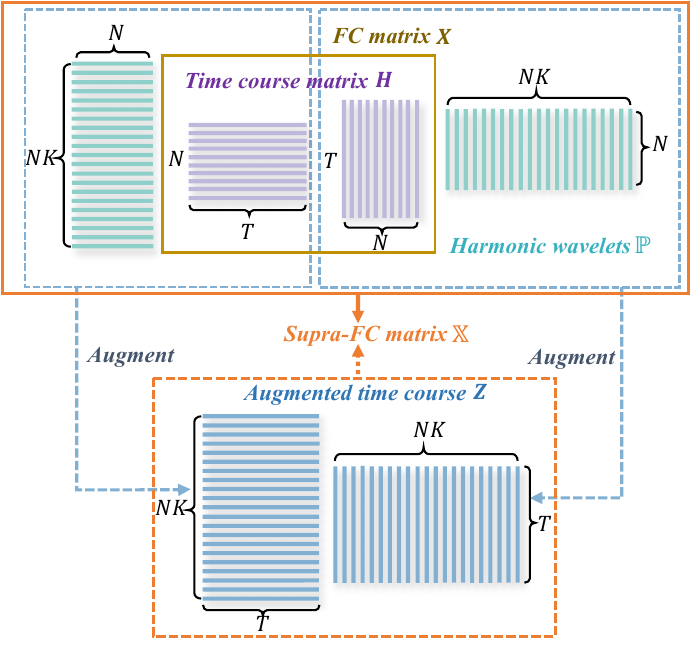}
    \vspace{-1.0em}
    \caption{Two ways to compute \textit{Supra-FC} matrix $\mathbb{X}$. Top: We apply scattering transforms $\mathbb{P}$ on FC matrix $\boldsymbol{X}$ (brown box) and obtain \textit{Supra-FC} matrix by $\mathbb{X}=\mathbb{P}^\intercal\boldsymbol{X}\mathbb{P}$. Bottom: An alternative way is to break the FC matrix into the covariance of time course matrix $\boldsymbol{H}$ (purple lines in the brown box). Then the inner product between each scattering transform and the snapshot of time course (columns in $\boldsymbol{H}$ is essentially the result of modulating the whole-brain signal of neural activity using harmonic wavelets from SC, which is constrained by the geometry of brain anatomy. Eventually, the correlation between augmented neural signals $\boldsymbol{Z}$ (blue lines) results in the same \textit{Supra-FC} matrix $\mathbb{X}$, accompanied by an in-depth neuroscience underpinning for SC-FC coupling mechanisms. }
    \label{snapshot}
\end{figure}

In Proposition \ref{SPD2}, we apply harmonic wavelets on FC matrix $\boldsymbol{X}$. Here, we break down FC matrix $\boldsymbol{X}$ into $\boldsymbol{HH}^\intercal$. Furthermore, we study each column in $\boldsymbol{H}$, which is the whole-brain snapshot of neural activities $h^t\in \mathbb{R}^N$ at time $t$. As the blue dashed box in Fig. \ref{snapshot}, we first apply harmonic wavelets $\mathbb{P}$ to each snapshot $h^t$, yielding an augmented time course matrix $\boldsymbol{Z}=\mathbb{P}\boldsymbol{H}\in \mathbb{R}^{NK\times T}$ (blue lines, compared to original signals $\boldsymbol{H}$). After that, the inner project of $\boldsymbol{Z}$ and $\boldsymbol{Z}^\intercal$ results in the \textit{Supra-FC} matrix $\mathbb{X}$ (orange dashed box).

\textbf{Neuroscience Insights.} The oscillation patterns of each harmonic wavelet $\psi_i^k$, constrained by the local topology of the structural connectome, characterize the frequency-specific neural activities supported by the underlying neural circuit. The inner project $\langle\psi_i^k,h^t\rangle$ over time essentially allows us to modulate the observed neural activity signals with the pre-define bandpass filters, which gives rise to coupled neural oscillations at distinct frequencies. In analogy to the holography technique in computer vision, the physics insight behind is that $\mathbb{X}=\boldsymbol{ZZ}^\intercal$ (orange dashed box in Fig. \ref{snapshot}) records interference patterns generated by two SC-modulated neural activity signals.

Since the geometry of harmonic wavelets is governed by the topology from SC, we have initialed the effort of building an integrated framework to analyze the SC-FC coupling mechanism using well-studied interference principles, which yields a new research paradigm using deep learning technology, called \textit{DeepHoloBrain}.

\textbf{Potential Applications.} \textit{First}, our manifold-based deep model can be used to investigate SC-FC coupling mechanisms in response to specific tasks or stimuli. Specifically, the input to our model is time course from task-fMRI images. The driving factor of machine learning is to predict the underlying cognitive tasks based on the learned intrinsic feature representation of functional fluctuations that describe the latent brain states. In addition, we integrate node-specific and frequency-specific attention components on top of the scattering transforms $\mathbb{P}$, which allows us to uncover how task-dependent spontaneous FC is supported by a collection of anatomical neural circuits in SC.

\textit{Second}, our approach offers an integrated approach to predict disease risks using both SC and FC information. Current machine learning methods primarily focus on feature fusion \cite{sarwar2021structure}. We model the SC-FC relationship through scattering transforms with great mathematical insight. Thus, our model offers in-depth network neuroscience underpinning to uncover the synergistic effect of brain structure/function and elucidate their mechanistic role in modifying cognitive performance in the context of neurological diseases.

\section{Experiments}

\subsection{Data Description and Experimental Setting}

\paragraph{Data Description.} To thoroughly validate the effectiveness of our proposed model, we utilize three existing public datasets for our experiments.

\textbf{The Lifespan Human Connectome Project Aging (HCP-A) dataset \cite{bookheimer2019lifespan}. } HCP-A is instrumental in task recognition research, offering a comprehensive view of the aging process. It includes data from 717 subjects, encompassing both fMRI (4,846 time series) and DWI (717) scans. This rich collection facilitates in-depth analyses of both functional and structural connectivity. HCP-A dataset includes data from four brain tasks associated with memory: VISMOTOR, CARIT, FACENAME, and Resting State. In the related experiments, these tasks are treated as distinct categories in a four-class classification problem.

\textbf{Alzheimer’s Disease Neuroimaging Initiative (ADNI) dataset \cite{weiner2015impact}}. The ADNI dataset serves as an invaluable resource, featuring a collection of 250 fMRI time series and 1,012 meticulously processed structural connectomes. In our experiment, we selectively focused on 250 subjects who had undergone both DWI and fMRI scans. This careful selection ensured that we had access to detailed data on both functional connectivity and structural connectomes for each subject. Additionally, the ADNI data includes clinical diagnostic labels, encompassing a spectrum of cognitive states: Cognitive Normal (CN), Subjective Memory Complaints (SMC), Early-Stage Mild Cognitive Impairment (EMCI), Late-Stage Mild Cognitive Impairment (LMCI), and Alzheimer's Disease (AD). Considering the data balance issue, we simplified these categories into two broad groups based on disease severity. Specifically, we combined CN, SMC, and EMCI into a single `CN' group, representing less severe conditions, while LMCI and AD were grouped together as the `AD' group. This categorization is intended to facilitate a binary classification framework for our analysis.

\textbf{Open Access Series of Imaging Studies (OASIS) dataset \cite{lamontagne2019oasis}}. The OASIS dataset presents a substantial collection of data from 924 subjects, comprising 3,322 fMRI sessions in total. Each subject has BOLD signals and structural connectome data, forming a comprehensive foundation for our analysis. In our experiment, we focused on binary classification: subjects in preclinical stages of Alzheimer's disease (categorized as preclinical-AD) or those manifesting dementia-related conditions (categorized as AD), while healthy individuals were classified under the `CN' group. This classification strategy aids in diagnosing the preclinical conditions at the early stage.

We summarize the processing steps for SC and FC data in the Appendix \ref{process}. In all of the following experiments, we partition each into 90 regions using AAL atlas \cite{tzourio2002automated}. Thus, SC is a $90\times 90$ matrix where each element is quantified by the number of fibers linking two brain regions. We further normalize the fiber count by the total fiber count of the underlying subject. Regarding FC, we first calculate the mean time course of BOLD signals in each region from fMRI. Then FC is a $90\times 90$ matrix by Pearson's correlation between mean time courses in different regions, which indicates temporal synchronization of neuronal activity in the brain (as shown in Fig. \ref{sc-fc}).

\textbf{Experimental Setup.} In our evaluations, we have conducted a thorough evaluation of the proposed \textit{DeepHoloBrain} model, focusing on two key classification challenges in neuroscience: brain task recognition and disease diagnosis. Since the geometric patterns in FC matrix have been widely investigated in the neuroscience field \cite{you2021re}, our comparative analysis includes a range of state-of-the-art SPD matrix learning methods: Covariance Discriminative Learning (CDL) \cite{wang2012covariance}, Log-Euclidean Metric Learning (LEML) \cite{huang2015log}, SPD Manifold Learning (SPDML) \cite{harandi2014manifold}, Affine-Invariant Metric (AIM) \cite{pennec2006riemannian}, Riemannian Sparse Representation (RSR) \cite{harandi2012sparse}, DeepO2P network \cite{ionescu2015matrix}, and SPDNet \cite{huang2017riemannian}. For each method, we utilize the source codes and meticulously fine-tuned parameters according to the specifications provided by their respective authors to ensure the integrity of our evaluation. The effectiveness of these comparative methods has been rigorously gauged using a comprehensive set of metrics, including accuracy, recall, F1-score, and precision. For all classification tasks, we use a cross-entropy loss function that compares the ground truth $G_{m,c}$ for the $c^{th}$ class of the $m^{th}$ sample against the corresponding prediction $\hat G_{m,c}$. Regarding our deep model, a regularization term is added to the scale parameter $\gamma$ in harmonic wavelets (Eq. \ref{eq.wavelets}) to prevent it from taking negative values:
\begin{equation}
\label{loss}
\ell (\gamma) = -\frac{1}{M} \sum_{m=1}^M \sum_{c \in C} G_{m, c} \log \hat{G}_{m, c}(\gamma) + \beta |\gamma| [\gamma<0]
\end{equation}
Here, $\beta$ is a hyperparameter, $M$ represents the sample size, and $|\gamma|[\gamma<0]$ denotes the absolute value of $\gamma$ when it is less than zero. We update the scale parameter using the rule $\gamma \gets \gamma-\alpha_\gamma \frac{\partial\ell}{\partial \gamma}$, where $\alpha_\gamma$ is the learning rate for $\gamma$. This update is performed via gradient-based methods along with other learnable parameters. For gradient back-propagation, we employ the Riemannian manifold optimizer \cite{huang2017riemannian}, with a learning rate of 0.001 and over 1000 epochs, the batch size is set to 16, weight decay is set to 1e-5, momentum is set as 0.9, the transformation matrix of SPDNet is initialized to a semi-positive definite matrix ($N \times 64$). Results from each experimental setup are reported using five-fold cross-validation, with both the mean and standard deviation documented for thorough analysis. The permutation t-test (the number of permutations is set to 10000) is used for statistical tests. All experiments were conducted on an Intel(R) Xeon(R) Gold 6448Y, paired with an NVIDIA RTX 6000 GPU.

\subsection{Performance on Task-specific Recognition}
In this series of experiments, we demonstrate task-specific recognition using several methods including SPDNet, CDL, LEML, SPDML, AIM, RSR, DeepO2P, and our \textit{DeepHoloBrain} (denoted as `OURS') on the HCP-A dataset. The numerical results, presented in Table \ref{task-reco}, clearly indicate that \textit{DeepHoloBrain} significantly outperforms (indicated by `*') the other seven methods, at a significance level $p<0.05$. Notably, all methods exhibit remarkable performance in this task. This can be attributed to the distinctively designed tasks (VISMOTOR, CARIT, FACENAME, and Resting State) within the dataset. However, the transformation matrix $\mathbf{\Psi}$ in Eq. \ref{eq:SPDNet}, a crucial component in these methods, is learned through a `black-box' deep learning approach, leaving the exact nature of the parameters learned somewhat obscure. To gain a deeper understanding of brain dynamics, we have developed two attention mappings: one for node-wise analysis and the other for graph wavelet construction (frequency). After retraining the model to focus solely on the VISMOTOR and FACENAME tasks, the visual results of the node-wise attention map, as shown in Fig. \ref{tak_recon2} (left), reveal a distinct spatial pattern in the sensorimotor and visual regions (color bar denotes the degree). This indicates that anatomical regions within the areas detected by our model are highly relevant to the VISMOTOR and FACENAME tasks, respectively. Such findings align with current research \cite{bedel2023bolt,glasser2016multi}, suggesting that our model is capable of capturing intrinsic features that differentiate various brain tasks. Furthermore, Fig. \ref{tak_recon2} (right) displays the oscillation patterns underlying the learned most relevant frequencies, indicating a distinct preference for high frequencies in the sensorimotor regions and lower frequencies within the visual region. This informative location-frequency attention map suggests that different cognitive tasks engage distinct neural circuits, manifesting as varied frequency responses across brain regions.

\begin{table}[h]
\vspace{-1.0em}
\caption{Results on brain task recognition for HCP-Aging dataset.}
\label{task-reco}
\vskip 0.15in
\begin{center}
\begin{small}
\begin{sc}
\setlength{\tabcolsep}{3pt}
\begin{tabular}{lcccr}
\toprule
Methods & Accuracy & Recall & F1-score \\ 
\midrule
SPDNet    & 0.984$\pm$ 0.003& 0.975$\pm$ 0.004& 0.978$\pm$ 0.004 \\
CDL & 0.976$\pm$ 0.003& 0.962$\pm$ 0.005& 0.966$\pm$ 0.005\\
LEML    & 0.961$\pm$ 0.022& 0.903$\pm$ 0.039& 0.929$\pm$ 0.036 \\
SPDML     & 0.944$\pm$ 0.015& 0.908$\pm$ 0.027&0.920$\pm$0.019 \\
AIM     & 0.952$\pm$ 0.014& 0.9114$\pm$ 0.016&0.929$\pm$0.015\\
RSR     & 0.966$\pm$ 0.005& 0.944$\pm$ 0.010& 0.951 $\pm$0.008 \\
DeepO2P      & 0.977$\pm$ 0.004& 0.963$\pm$ 0.006& 0.969$\pm$ 0.005 \\
OURS   &\bf{0.995$\pm$ 0.003*}& \bf{0.989$\pm$ 0.003*}& \bf{0.993$\pm$ 0.003*} \\
\bottomrule
\end{tabular}
\end{sc}
\end{small}
\end{center}
\vspace{-1.0em}
\end{table}

\begin{figure}[h]
    \centering
    \includegraphics[width=0.48\textwidth]{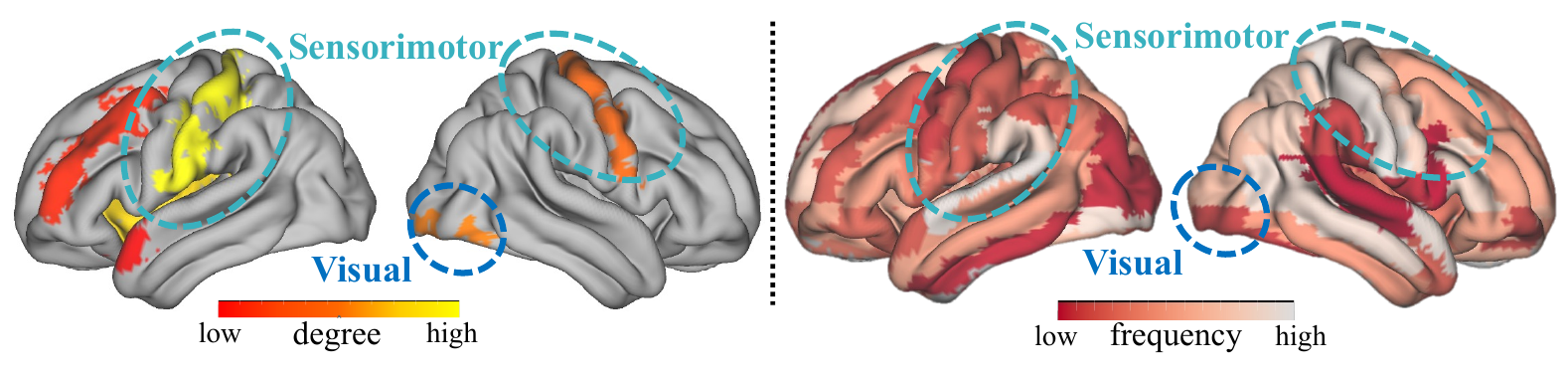}
    \vspace{-2.0em}
    \caption{Left: Uncovered anatomical regions (using node-wise attention) for task VISMOTOR and task FACENAME, respectively. Right: Oscillation patterns underlying the more relevant harmonic frequencies for VISMOTOR and FACENAME tasks, revealing task-specific wavelet dynamics.}
    \vspace{-1.0em}
    \label{tak_recon2}
\end{figure}

\subsection{Performance on Disease Diagnosis}
\textbf{Disease Early Diagnosis.} We evaluate the prediction accuracy on ADNI and OASIS data, respectively, in forecasting the risk of developing AD by SPDNet, CDL, LEML, SPDML, AIM, RSR, DeepO2P and \textit{DeepHoloBrain} (marked `OURS') in Table \ref{AD-dig} (top and middle). At the significance level of 0.001, our method outperforms all other counterpart methods in terms of classification accuracy (indicated by `*').

\textbf{Altered SC-FC Coupling in AD.} Here, we seek to identify the alterations of SC-FC coupling that underline the progression of AD. In this regard, we apply the node-size attention module to characterize the regional contribution of SC-FC coupling in recognizing AD from CN, which implies the most vulnerable brain regions in AD progression. As illustrated in Fig. \ref{ad-dis}, the analysis on both OASIS and ADNI datasets reveals that the regions predominantly affected are within the default mode regions (highlighted by red circles). This finding aligns with numerous studies that have established a strong correlation between these regions and human memory \cite{greicius2004default, koch2012diagnostic}.

\textbf{Performance on Learning Consistency across Data Cohorts.} Herein, we focus on consistency in terms of learned attention maps from ADNI and OASIS. Although both are AD-specific studies, variations exist in the choice of scanner, study design, and demographic data of the subjects between the two studies. Particularly, we skip the data harmonization routine \cite{richter2022validation} and train our deep model on ADNI and OASIS separately. The spatial maps of region-specific and frequency-specific attentions are shown in Fig. \ref{ad-dis} side by side. It is evident that the oscillation patterns in frequency-specific attention are quite consistent between ADNI and OASIS. Although we only find one region (circles in Fig. \ref{ad-dis} bottom) consistently selected in two data cohorts, the top-ranked regions in OASIS and ADNI (circles in Fig. \ref{ad-dis} top) are both located in the default-mode network \cite{greicius2004default,lee2020relationship,mevel2011default}, a sub-network frequently affected by AD.
\begin{figure}[h]
    \centering
    \vspace{-0.5em}
        \includegraphics[width=0.48\textwidth]{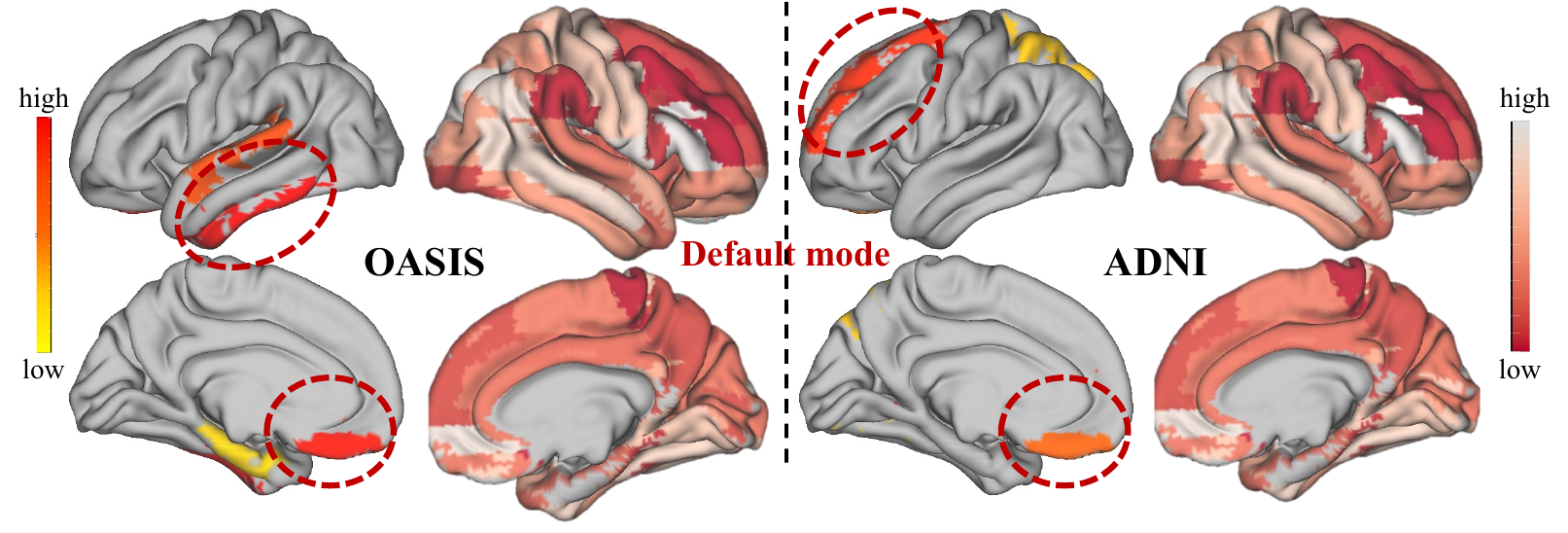}
   \vspace{-2.5em}
    \caption{Consistency evaluation for region-specific ($1^{st}$ and $3^{rd}$ columns) and frequency-specific ($2^{nd}$ and $4^{th}$ columns) attentions learned from OASIS (left) and ADNI (right).}
    \vspace{-1.0em}
    \label{ad-dis}
\end{figure}

\begin{table}[h]
\vspace{-1.0em}
\caption{Prognosis accuracies on forecasting AD risk for OASIS (top) and ADNI (shaded, middle) datasets. Bold font denotes the best performance. `OURS' refers to our proposed model, named \textit{DeepHoloBrain}. When we mention `baseline' (darker shade), they specifically denote the version of the models that have been trained using a particular dataset (i.e., OASIS and ADNI). Additionally, `baseline+' (lighter shade, bottom) represents an enhanced variant of the models, which has been pre-trained on the HCP-A dataset and subsequently fine-tuned on the ADNI dataset.}
\label{AD-dig}
\vskip 0.15in
\begin{center}
\begin{small}
\begin{sc}
\setlength{\tabcolsep}{2pt}
\begin{tabular}{lccccc}
\toprule
Methods & Accuracy & Recall & F1-score \\ 
\midrule
SPDNet    & 0.871$\pm$ 0.018& 0.593$\pm$ 0.017& 0.613$\pm$ 0.024  \\
CDL & 0.827$\pm$ 0.623& 0.570$\pm$ 0.062& 0.581 $\pm$0.047 \\
LEML    & 0.796$\pm$ 0.193& 0.624$\pm$ 0.045&0.632$\pm$ 0.117 \\
SPDML     & 0.786$\pm$ 0.148& 0.739$\pm$ 0.038& 0.697$\pm$0.091        \\
AIM     & 0.816$\pm$ 0.098& 0.540$\pm$ 0.043& 0.539$\pm$0.072\\
RSR     & 0.840$\pm$ 0.020& 0.690$\pm$ 0.020& 0.689 $\pm$0.028\\
DeepO2P      & 0.857 $\pm$ 0.019& 0.684$\pm$ 0.02& 0.675$\pm$ 0.021  \\
OURS  & \bf{0.885$\pm$0.017*}& \bf 0.740$\pm$ 0.045*& \bf 0.6974$\pm$ 0.041* \\
\hline
\hline
\rowcolor{gray!40} SPDNet    & 0.800$\pm$0.085& 0.670$\pm$ 0.067& 0.627$\pm$ 0.090 \\
\rowcolor{gray!40} CDL & 0.710$\pm$0.095& 0.500$\pm$ 0.018& 0.415$\pm$0.064 \\
\rowcolor{gray!40} LEML    & 0.704$\pm$0.095& 0.523$\pm$0.019& 0.474$\pm$0.065 \\
\rowcolor{gray!40} SPDML     & 0.672$\pm$ 0.079& 0.543$\pm$ 0.037& 0.529$\pm$ 0.055        \\
\rowcolor{gray!40} AIM     & 0.708$\pm$ 0.089& 0.500$\pm$ 0.000& 0.413$\pm$ 0.032 \\
\rowcolor{gray!40} RSR     & 0.740$\pm$ 0.106& 0.610$\pm$ 0.059& 0.608$\pm$ 0.080 \\
\rowcolor{gray!40} DeepO2P     & 0.760$\pm$ 0.089& 0.614$\pm$ 0.068& 0.625$\pm$0.082         \\
\rowcolor{gray!40} OURS  & \bf{0.820$\pm$ 0.071*}& \bf{0.625$\pm$ 0.049*}& \bf{0.647$\pm$0.079*} \\
\midrule
\rowcolor{gray!10}
\multicolumn{4}{@{}l@{}}{\makebox[0.72\linewidth][l]{~~~~~ \hspace{8pt}  ~ ~  ~~~Accuracy \hspace{2pt} Recall \hspace{2pt} F1-score \hspace{2pt} Precision}} \\
\rowcolor{gray!10}
\multicolumn{4}{@{}l@{}}{\makebox[0.72\linewidth][l]{SPDNet+ \hspace{1pt} ~~ 0.7120 \hspace{10pt} {\bf0.7120} \hspace{12pt} 0.6632 \hspace{13pt} 0.7409}} \\
\rowcolor{gray!10}
\multicolumn{4}{@{}l@{}}{\makebox[0.72\linewidth][l]{CDL+ \hspace{1pt} ~ ~ ~ ~~~0.6640 \hspace{10pt} 0.5253 \hspace{12pt} 0.5168 \hspace{12pt} 0.5397}} \\
\rowcolor{gray!10}
\multicolumn{4}{@{}l@{}}{\makebox[0.72\linewidth][l]{LEML+ \hspace{1pt} ~ ~~~ 0.6080 \hspace{10pt} 0.5501 \hspace{12pt}  0.5466 \hspace{12pt} 0.5462}} \\
\rowcolor{gray!10}
\multicolumn{4}{@{}l@{}}{\makebox[0.72\linewidth][l]{SPDML+ \hspace{1pt} ~~ 0.5880 \hspace{11pt}~0.5803 \hspace{12pt} 0.5581 \hspace{12pt} 0.5671}} \\
\rowcolor{gray!10}
\multicolumn{4}{@{}l@{}}{\makebox[0.72\linewidth][l]{AIM+ \hspace{1pt} ~ ~ ~ ~~ 0.6160 \hspace{10pt} 0.5356 \hspace{12pt} 0.5356 \hspace{12pt} 0.5356}} \\
\rowcolor{gray!10}
\multicolumn{4}{@{}l@{}}{\makebox[0.72\linewidth][l]{RSR+ \hspace{1pt} ~ ~ ~ ~~ 0.6880 \hspace{10pt} 0.6880 \hspace{12pt} ~{\bf 0.8690} \hspace{13pt} 0.4716}} \\
\rowcolor{gray!10}
\multicolumn{4}{@{}l@{}}{\makebox[0.72\linewidth][l]{DeepO2P+ \hspace{1pt} 0.6960 \hspace{10pt} 0.6880 \hspace{12pt}~ 0.5971 \hspace{12pt} 0.7369}} \\
\rowcolor{gray!10}
\multicolumn{4}{@{}l@{}}{\makebox[0.72\linewidth][l]{OURS+ \hspace{1pt} ~ ~~~ {\bf{{0.7400*}}} \hspace{8pt} 0.6103\hspace{18pt}0.6081 \hspace{12pt} \bf 0.7892*}} \\
\bottomrule
\end{tabular}
\end{sc}
\end{small}
\end{center}
\vspace{-1.5em}
\end{table}

\textbf{Generality as A Pre-trained Model.} One of the critical challenges of deploying computer-assisted diagnosis in clinical routine is the limited sample size, especially for disease cohorts. Considering the fact that (1) the pre-defined scattering transforms underscore the coupling mechanism of SC-FC relationship and (2) the \textit{MLP-Mixer} architecture enables us to train the deep model with much fewer parameters than conventional methods, we hypothesize that our model pre-trained on a large dataset of normal healthy brains can produce reasonable classification results for disease diagnosis after fine-tuning on a limited amount disease-specific data. To test this hypothesis, we pre-train a regression model based on Montreal Cognitive Assessment (MoCA) score on HCP-A data and fine-tune a classification mode on ADNI data.
The experimental results, as presented in Table \ref{AD-dig} (bottom), it is apparent that our method has achieved significant improvement in accuracy and precision, with a marginal decrease in recall and F1-score. Furthermore, we compute the AUC for SPDNet and our \textit{DeepHoloBrain}, where our method achieves 1.76\% improvement over SPDNet. Considering early diagnosis of AD is critical in developing interventional therapeutics, the results in Table 2 (bottom) suggest our method has better clinical value than SPDNet in terms of diagnosis sensitivity.

\section{Discussion}
In this section, we delve into the performance comparison of task recognition and disease diagnosis using graph-based and transformer-based methods, which have been extensively employed in neuroscience research. Specifically, we evaluate six methods, comprising four GCN-based approaches and two transformer models. The comparison encompasses GCN \cite{kipf2017semisupervised}, GAT \cite{velivckovic2017graph}, GIN \cite{xu2018powerful}, BrainNetCNN \cite{kawahara2017brainnetcnn}, Graph Transformer (GraphT) \cite{shi2020masked}, and BolT \cite{bedel2023bolt}, assessed on both the HCP-A dataset (top) and an ADNI dataset (bottom) in Table \ref{graph-transformer}. Notably, for GCN-based methods, we utilize FC as the feature representation underlying the SC topology. For BolT and BrainNetCNN, only FC information is used in the deep model.

\begin{table}[h]
\vspace{-1.5em}
\caption{Performance on task recognition for HCP-A dataset (top) and forecasting AD risk for ADNI datasets (bottom), respectively. }
\label{graph-transformer}
\vskip 0.15in
\begin{center}
\begin{small}
\setlength{\tabcolsep}{2pt}
\begin{tabular}{lccccc}
\toprule
Methods & Accuracy & Precision & F1-score \\ 
\midrule
GCN    & 0.690$\pm$ 0.057& 0.511$\pm$ 0.157& 0.582$\pm$ 0.116 \\
GAT & 0.719$\pm$ 0.123&  0.540$\pm$ 0.222 & 0.611 $\pm$0.182 \\
GIN    & 0.585$\pm$ 0.124& 0.513$\pm$ 0.162&0.694$\pm$ 0.066 \\
GraphT     & 0.721$\pm$ 0.127& 0.542$\pm$ 0.227& 0.614$\pm$0.189 \\
BrainNetCNN    & 0.725$\pm$ 0.135& 0.546$\pm$ 0.235 & 0.617$\pm$0.195\\
{BolT}    & 0.863$\pm$ 0.013& 0.865$\pm$ 0.015& 0.864 $\pm$0.014\\
\hline
\hline
\rowcolor{gray!40} GCN    & 0.541$\pm$0.313& 0.447$\pm$ 0.369& 0.472$\pm$ 0.366  \\
\rowcolor{gray!40} GAT & 0.659$\pm$0.287& 0.562$\pm$ 0.319& 0.596$\pm$0.319 \\
\rowcolor{gray!40} GIN    & 0.600$\pm$0.106& 0.728$\pm$0.073& 0.642$\pm$0.093 \\
\rowcolor{gray!40} GraphT     & 0.541$\pm$ 0.275& 0.595$\pm$ 0.297& 0.497$\pm$ 0.276        \\
\rowcolor{gray!40} BrainNetCNN     & 0.800$\pm$ 0.085& 0.638$\pm$ 0.124& 0.643$\pm$ 0.102 \\
\rowcolor{gray!40} BolT     & 0.812$\pm$ 0.075& 0.618$\pm$ 0.142& 0.645$\pm$ 0.105 \\
\bottomrule
\end{tabular}
\end{small}
\end{center}
\vspace{-2.0em}
\end{table}

It is evident that GCN-based methods yield less satisfactory results in recognizing cognitive tasks for normal adults (in HCP-A), largely due to the neuroscience fact that the topological pattern usually does not highly correlate with functional fluctuations exhibited in cognitive tasks. Meanwhile, some GCN-based methods are able to yield reasonable prediction accuracy of AD risk since SC alteration is a hallmark of AD progression in clinical findings. Taken together, these results suggest that effectively employing GCN-based techniques from machine learning in neuroscience requires tailored adaptation of domain-specific expertise. Following this spirit, our \textit{DeepHoloBrain} integrates physics principles, neuroscience insights, and the power of machine learning, which allows us to consistently outperform existing deep models. A good integration of SC and FC information effectively provides a greater advantage than using SC or FC alone. For instance, although BolT uses a more advanced machine learning backbone (Transformer) than our \textit{DeepHoloBrain} (MLPs), BolT (using FC only) shows less accuracy in both cognitive task recognition (0.863 vs. 0.995) and disease diagnosis (0.812 vs. 0.820) than our method (using both SC and FC).

\section{Conclusions}
In this work, we have ventured into developing a unique deep learning framework, \textit{DeepHoloBrain}, which utilizes a scattering-transform mixer on the Riemannian manifold to delve into the complexities of neural dynamics. At its core, \textit{DeepHoloBrain} is a physics-informed model that aims to unravel the intricate coupling mechanism between brain structure and function, viewed through the prism of data geometry. Our model innovatively combines: (1) A new mapping mechanism on the Riemannian manifold using scattering transforms; (2) A pioneering geometric deep learning structure, tailored for SPD matrices; (3) An insightful approach to understanding the human brain with deep learning; (4) A thorough evaluation against large-scale neuroimaging data, proving its effectiveness and practicality. Our extensive experimental results affirm the model’s effectiveness and practicality, marking a promising pilot work in the field of neuroscience.


\section*{Software and Data}

HCP-A data can be found in \url{https://www.humanconnectome.org/study/hcp-lifespan-aging}, ADNI can be downloaded in \url{https://adni.loni.usc.edu/}, OASIS can be downloaded in \url{https://sites.wustl.edu/oasisbrains/home/oasis-3/}. The code is released in GitHub: \url{https://github.com/Dandy5721/ICML2024.git}.

\section*{Acknowledgements}

This work was supported by the National Institutes of Health AG070701, AG073927, AG068399, and Foundation of Hope.

\section*{Impact Statement}

This paper advances Machine Learning and Neuroscience by introducing novel concepts and methods. Our major contributions to the machine learning field are we propose (1) a new mapping mechanism on the Riemannian manifold, constrained by scattering transforms, and (2) a novel geometric deep model - Scattering-transform Mixer for the Riemannian manifold of SPD matrices. From an application standpoint, it offers (1) a novel approach to understanding the human brain via deep learning techniques, and (2) a comprehensive evaluation of these techniques' clinical impact on large-scale neuroimaging datasets. These advancements promise significant implications in both theoretical and practical domains of neural science and machine learning. Potential applications can be referred to the last paragraph of Section \ref{DeepHoloBrain}.

\nocite{langley00}

\bibliography{example_paper}
\bibliographystyle{icml2024}

\newpage
\appendix
\onecolumn
\section{Appendix}

\subsection{Construction of SC and FC Networks}
\label{process}
The data preprocessing involves the following steps to derive FC and SC matrices:

$\bullet$ \textbf{fMRI Data Processing\footnote{\url{https://fmriprep.org/en/stable/}}:}

$\triangleright$ \textit{Structural MRI (T1-weighted) preprocessing}, including brain extraction, tissue segmentation, spatial normalization, cost function masking, longitudinal processing, and brain mask refinement.

$\triangleright$ \textit{BOLD preprocessing}, encompassing reference image estimation, head-motion estimation, slice time correction, susceptibility distortion correction, EPI to T1w registration, resampling into standard spaces, sampling to Freesurfer surfaces, HCP Grayordinates, and confounds estimation.

$\triangleright$ \textit{Generation of BOLD time-series data for FC matrices construction.}

$\bullet$ \textbf{DWI Data Processing\footnote{\url{https://qsiprep.readthedocs.io/en/latest/}}:}

$\triangleright$ \textit{Initial processing}, including image and gradient orientation conformity, distortion grouping, denoising, distortion correction, head motion correction, and B0 template construction.

$\triangleright$ \textit{Reconstruction steps}, such as ODF/FOD estimation, anisotropy scalar computation, and tractography.

$\triangleright$ \textit{Estimation of structural connectomes to generate SC matrices.}

Figure \ref{sc-fc} illustrates the process of constructing FC and SC networks.

\begin{figure}[h]
    \centering
    \vspace{-1.0em}
        \includegraphics[width=0.98\textwidth]{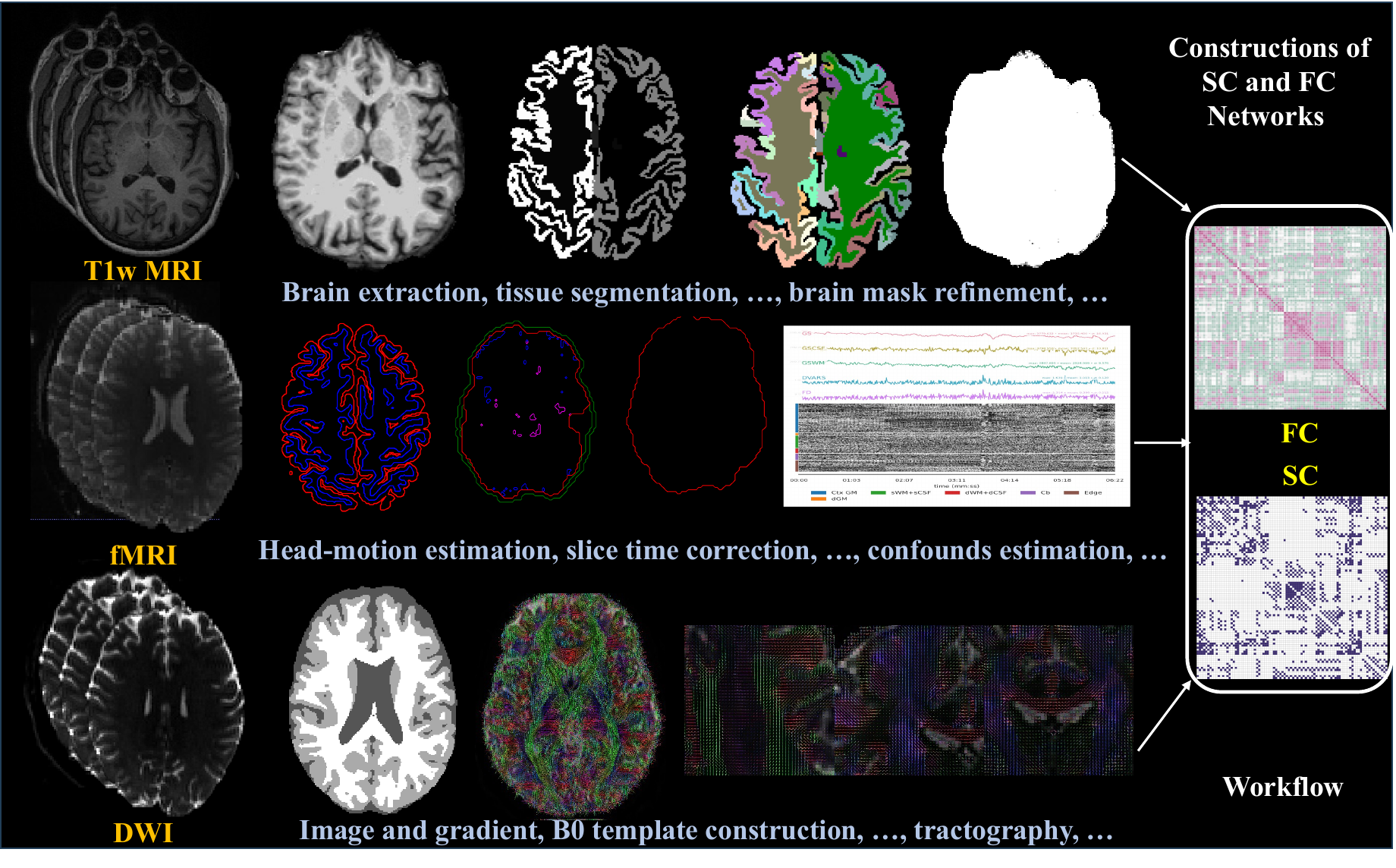}
    \caption{Workflow of constructing FC and SC.}
    \label{sc-fc}
\end{figure}

\subsection{Proof}
\begin{proposition}
    Given a full-rank matrix $ W \in \mathbb{R}^{n \times m} $ and a vector $ v \in \mathbb{R}^m $, since the columns of $ W $ are linearly independent and span $ \mathbb{R}^m $, there exists, by the fundamental theorem of linear algebra, a vector $ u \in \mathbb{R}^n $ such that $ W u = v $.

    \label{theorem:lowrank2}
\end{proposition}
\begin{proof}
Since $ W $ is full rank, its column space is a subspace of $ \mathbb{R}^m $. If the rank of $ W $ is $ m $, it means that its column space is actually the entire $ \mathbb{R}^m $. Thus, every vector in $ \mathbb{R}^m $ can be linearly represented by the columns of $ W $, that is, for each $ v \in \mathbb{R}^m $, there exist coefficients $ c_1, c_2, \ldots, c_m $ such that $v = c_1 w_1 + c_2 w_2 + \ldots + c_m w_m$, where $ w_1, w_2, \ldots, w_m $ are the columns of $ W $. We can write this equation to matrix form as

\begin{equation*}
v = W \begin{bmatrix}
c_1 \\
c_2 \\
\vdots \\
c_m
\end{bmatrix},
\end{equation*}

where the vector of coefficients is $ u \in \mathbb{R}^n $ which satisfies $ Wu = v $, proving the existence of such a vector $ u $ for every $ v \in \mathbb{R}^m $.
\end{proof}

\begin{proposition}
\label{laplacian}
\textit{Remark.} Suppose we have a graph Laplacian matrix $ \mathcal{L} = D - A $, where $ D $ is the degree matrix and $ A $ is the adjacency matrix, the first eigenvalue is typically $ \lambda_0 = 0 $, corresponding to the trivial eigenvector $ u_0 $. This is because the graph Laplacian represents the nodal connectivity, and the constant vector represents an equal state of all nodes, naturally belonging to the kernel of $ \mathcal{L} $.

The eigenvector $ u_0 $ is usually a vector with all equal elements. For an undirected graph, it is often chosen as $ u_0 = \frac{1}{\sqrt{N}}(1, 1, \ldots, 1)^\intercal $, where $ N $ is the number of nodes in the graph. Thus, the first vector $ u_0 $ has all positive elements and this choice satisfies the definition of an eigenvector:

\[ \mathcal{L} u_0 = (D - A) u_0 = 0 \cdot u_0 = 0 \]

    and the normalization condition of $ u_0 $ (i.e., $ \|u_0\| = 1 $), since each element is $\frac{1}{\sqrt{N}} $, the norm (length) of the vector is 1.

In this context, the value range harmonic wavelet is $[-\infty,\frac{1}{N},+\infty]$, thus the lower band of $max(\cdot)$ operator is $\frac{1}{{N}}>0$.

\end{proposition}

\begin{proposition}
\label{spd_prove1}
Given a symmetric matrix $A$,  $B=\exp(A)$ is a SPD matrix.
\end{proposition}
\begin{proof} To prove that $\exp(A)$ is positive definite, we use the fact that any symmetric matrix can be diagonalized. This means that there exists an orthogonal matrix $Q$ and a diagonal matrix $D$ such that $A = QDQ^\intercal$.

Then, we have:
\[ \exp(A) = \exp(QDQ^\intercal) = Q\exp(D)Q^\intercal \]
where the exponential of a diagonal matrix is obtained by applying the exponential function to its diagonal elements.

Consider any non-zero vector $v$, we have:
\[ v^\intercal \exp(A) v = v^\intercal Q\exp(D)Q^\intercal v \]
Let $y = Q^\intercal v$, then we get:
\[ v^\intercal \exp(A) v = y^\intercal \exp(D) y \]

Since $\exp(D)$ is positive definite (its diagonal elements are the results of the exponential function, hence all positive), we have:
\[ y^\intercal \exp(D) y > 0 \]
unless $y$ is a zero vector. However, since $Q$ is orthogonal, $y = 0$ if and only if $v = 0$. Therefore, for all non-zero $v$, the expression is positive.

This proves that $B=\exp(A)$ is a positive definite matrix.

\end{proof}

\begin{proposition}
\label{spd_prove2}
Given a positive definite matrix $A$, then $B=(A + A^\intercal)/2$ is a SPD matrix.
\end{proposition}
\begin{proof}
\textbf{1: Symmetry}:
   The symmetry of $ B $ can be shown quite straightforwardly. Given $ A $, a positive definite matrix, let $ A^\intercal $ denotes its transpose. The matrix $ B $ is defined as $ B = \frac{A + A^\intercal}{2} $. Since the transpose of a sum of matrices is the sum of their transposes, we have $ B^\intercal = \frac{A^\intercal + (A^\intercal)^\intercal}{2} $. Since $ (A^\intercal)^\intercal = A $, it follows that $ B^\intercal = B $. Hence, $ B $ is symmetric.

\textbf{2: Positive Definiteness}: For any non-zero vector $ v $, consider the quadratic form $ v^\intercal Bv $:
   \[ v^\intercal Bv = v^\intercal \left( \frac{A + A^\intercal}{2} \right) v \]
   \[ = \frac{1}{2} v^\intercal Av + \frac{1}{2} x^\intercal A^\intercal v \]
   Since $ A $ is positive definite, $ v^\intercal Av > 0 $ for all non-zero $ v $. Similarly, $ v^\intercal A^\intercal v = (Av)^\intercal v = v^\intercal Av > 0 $ as $ A $ is positive definite. Obviously, $ v^\intercal Bv = (\frac{1}{2} v^\intercal Av + \frac{1}{2} v^\intercal A^\intercal v) > 0$. Hence, $ B $ is positive definite.

\end{proof}


\end{document}